\documentclass[10pt,aps,pra,twocolumn,superscriptaddress,amsmath,amssymb,longbibliography,eqsecnum]{revtex4-2}
\usepackage{graphicx,physics}
\usepackage[usenames,dvipsnames,svgnames,table]{xcolor}
\usepackage[unicode=true, pdfusetitle, bookmarks=true, bookmarksnumbered=false, bookmarksopen=false, 
			pdfborder={0 0 0}, backref=false, colorlinks=true]{hyperref}
\hypersetup{linkcolor=NavyBlue,urlcolor=NavyBlue,citecolor=NavyBlue}
\usepackage{amsthm}
\usepackage{xspace}

\newcommand{\sect}[1]{\textit{#1}.---\!} 
\newcommand{\bigsect}[1]{}

\newcommand{\qfi}{\mathcal{F}}
\newcommand{\curse}{\mathcal{C}}

\newcommand{\sfA}{\mathsf{A}}
\newcommand{\sfB}{\mathsf{B}}
\newcommand{\sfC}{\mathsf{C}}
\newcommand{\sfD}{\mathsf{D}}
\newcommand{\sfM}{\mathsf{M}}
\newcommand{\sfN}{\mathsf{N}}
\newcommand{\sfK}{\mathsf{K}}
\newcommand{\sfX}{\mathsf{X}}
\newcommand{\sfY}{\mathsf{Y}}
\newcommand{\sfZ}{\mathsf{Z}}
\newcommand{\sfO}{\mathsf{O}}

\newcommand{\sfR}{\mathsf{R}}
\newcommand{\sfI}{\mathsf{I}}
\newcommand{\sfP}{\mathsf{P}}
\newcommand{\est}[1]{\hat{#1}}

\newtheorem{thm}{Theorem}
\newtheorem{cor}{Corollary}

\begin{document}

\title{Generalization of Rayleigh's Curse on Parameter Estimation with Incoherent Sources}
\author{Lijun Peng}
\affiliation{Department of Physics, Hangzhou Dianzi University, Hangzhou 310018, China}

\author{Xiao-Ming Lu}
\email{lxm@hdu.edu.cn}
\homepage{http://xmlu.me}
\affiliation{Department of Physics, Hangzhou Dianzi University, Hangzhou 310018, China}

\begin{abstract}
The basic idea behind Rayleigh's criterion on resolving two incoherent optical point sources is that the overlap between the spatial modes from different sources would reduce the estimation precision for the locations of the sources, dubbed Rayleigh's curse.
We generalize the concept of Rayleigh's curse to the abstract problems of quantum parameter estimation with incoherent sources.
To manifest the effect of Rayleigh's curse on quantum parameter estimation, we define the curse matrix in terms of quantum Fisher information and introduce the global and local immunity to the curse accordingly.
We further derive the expression for the curse matrix and give the necessary and sufficient condition on the immunity to Rayleigh's curse.
For estimating the one-dimensional location parameters with a common initial state, we demonstrate that the global immunity to the curse on quantum Fisher information is impossible for more than two sources.
\end{abstract}

\maketitle

\section{Introduction}
Rayleigh's criterion on distinguishing two incoherent optical point sources determines a  characteristic distance, below which resolving two incoherent sources by imaging are thought to be difficulty~\cite{LordRayleigh1879}.  
The basic idea behind Rayleigh's criterion is that the overlap between the spot images would impede the resolution of the point sources.
As noted in the Feynman lectures on physics~\cite[Sec.~30–4]{Feynman1963}, Rayleigh's criterion on the resolution of point sources of lights is a rough idea.
Instead, a meticulous method of measuring the resolution power of incoherent point sources can be achieved by resorting to statistical approaches, e.g., parameter estimation~\cite{Ram2006,Chao2016} or hypothesis testing~\cite{Helstrom1973,Harris1964,Acuna1997,Shahram2006}.

One great advantage of the parameter estimation approach as well as the hypothesis testing approach is the availability of optimizing over quantum measurements~\cite{Tsang2016b,Tsang2019a,Lu2018} using quantum estimation and detection theory~\cite{Helstrom1976,Holevo1982}.
The quantum limit of the estimation precision for the separation of the two incoherent point sources has been revealed to be unexpectedly much higher than that of direct imaging in the sub-Rayleigh region~\cite{Tsang2016b}, indicating that the overlap between quantum states may not impede the estimation of the separation, as the conventional wisdom suggests.
Tsang {\it et al.} introduced Rayleigh's curse in Ref.~\cite{Tsang2016b} to dub the effect that ``the positions of two incoherent sources should become harder to estimate when their radiations overlap.'' 
and showed that the estimation precision of the separation between the two optical point sources is immune to Rayleigh's curse when optimizing quantum measurements. 
A lot of efforts has been devoted into further investigating and demonstrating the quantum superiority brought by optimizing measurements~\cite{Pirandola2018,Nair2016,Yang2016,Paur2016,Tham2017,Nair2016a,Lupo2016,Rehacek2017a,Rehacek2017,Ang2017,Chrostowski2017,Dutton2019,Tsang2019,Zhou2019,Tsang2019d,Rehacek2019,Bonsma-Fisher2019,Napoli2019,Len2020,Datta2020}.

To further understand and systematically study the effect of the overlap between quantum states from different sources on the precision of parameter estimation, we in this work generalize Rayleigh's curse to the abstract problems of parameter estimation with incoherent sources. 
We consider the influence of Rayleigh's curse on the quantum Fisher information (QFI) matrix~\cite{Helstrom1976,Holevo1982,Liu2020}, which plays the pivotal role in quantum parameter estimation.
To manifest Rayleigh's curse on the QFI, we define the curse matrix by comparing the real QFI matrix and that as if the quantum states from different sources are orthogonal.
We further introduce the global and local immunity to the QFI curse and derive the necessary and sufficient condition on the immunity of the QFI curse.
For estimating one-dimensional location parameters with a common initial state, we show that the global immunity is impossible for more than two incoherent sources.
Moreover, for two-source cases, we derives an analytic form of the curse matrix.

\section{Parameter Estimation with Incoherent Sources}

Let us start by supposing a set of incoherent sources, each emitting a signal that is transmitted to a common receiver (see Fig.~\ref{fig:incoheret} for an illustration).
The signals are received in the form of quantum states.
Assume that the quantum states received from each source are pure states and denote them by a set of state vectors \(\ket{\psi_j}\) for \(j=1,2,\ldots n\).
Even when quantum state emitted by different sources are orthogonal, the received states may be non-orthogonal, due to some physical processes during the transmission and reception, e.g., diffraction of light, noise effects, and filtering operations.
These processes are represented by an abstract channel in Fig.~\ref{fig:incoheret}.
The incoherence of the sources is manifested in the fact that the relative phases between different state vectors \(\ket{\psi_j}\) are totally random.
As a result, the quantum state of the receiver is described by the density operator in the form of
\begin{equation} \label{eq:model}
	\rho = \sum_{j=1}^n w_j \op{\psi_j},
\end{equation}
where the weight \(w_j\) is the prior probability of the system being in the state \(\ket{\psi_j}\).

In this work, we consider parameter estimation problems with incoherent sources.
We assume that the state vectors \(\ket{\psi_j}\) depend on a \(d\)-dimensional vector parameter \( \theta = (\theta_1,\theta_2,\ldots,\theta_d)\) and the weights \(w_j\)'s s are independent of \(\theta\). 
The value of \(\theta\) need be estimated by observing the quantum system possessed by the receiver. 
Particularly, we consider the cases where the number of incoherent sources is much less than the dimension of the Hilbert space for the underlying quantum system.
In other words, the density operators considered in this work are of low rank.
This includes the special case of infinite dimensional Hilbert space.

\begin{figure}[tb]
	\includegraphics[]{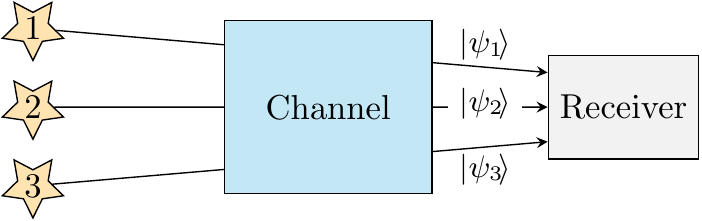}
	\caption{\label{fig:incoheret} Incoherent source model.}
\end{figure}

Quantum parameter estimation theory provides a powerful tool for pursuing the fundamental limits on the estimation error. 
For any unbiased estimator and any quantum measurement, the error-covariance matrix \(\mathcal E\) of parameter estimation, defined by \(\mathcal E_{jk} = \mathbb E[(\est\theta_j-\theta_j)(\est\theta_k-\theta_k)]\) with \(\hat\theta_j\) being the estimate for \(\theta_j\) and \(\mathbb E\) the expectation over the observations, obeys the quantum Cram\'er-Rao bound (QCRB) 
\begin{equation}
	\mathcal{E} \ge \qfi^{-1},
\end{equation} 
where \(\qfi\) is the QFI matrix~\cite{Helstrom1976,Holevo1982,Helstrom1967,Helstrom1968,Paris2009,Liu2020}.
The entries of the QFI matrix are defined by 
\begin{equation}
	\qfi_{\mu\nu} = \Re\tr(L_\mu L_\nu \rho), 	
\end{equation}
where the symmetric logarithmic derivative (SLD) operator \(L_\mu\) for \(\theta_\mu\) is an Hermitian operator satisfying 
\begin{equation}
	\frac12(L_\mu \rho + \rho L_\mu) = \partial_\mu \rho
\end{equation}
with \(\partial_\mu:=\partial/\partial\theta_\mu\) being defined for short.
Despite that the QCRB is not guaranteed to be attainable in general for the joint estimation of multiple parameters, the QFI matrix still reflects lots of information about the quantum limit of estimation errors~\cite{Belavkin1976,Lu2020b,Gross2020,Xing2020,Carollo2019a,Carollo2020,Miyazaki2020}.

\sect{Gauge symmetry}
Before calculating the QFI matrix, observe that, due to the incoherent characteristic of the model, the density operator in Eq.~\eqref{eq:model} is invariant under the (local) gauge transformations \(\ket{\psi_j} \mapsto e^{i\alpha_j} \ket{\psi_j}\) for \(j=1,2,\ldots,n\), where \(\alpha_j\) are arbitrary real-valued functions of \(\theta\).
The QFI matrix is solely determined by the state vectors \(\ket{\psi_j}\) and their derivatives \(\partial_\mu\ket{\psi_j}\);
The latter will have additional terms under the local gauge transformations, that is,
\begin{equation}
	\partial_\mu \ket{\psi} 
	\mapsto \partial_\mu e^{i\alpha} \ket{\psi}
	= e^{i\alpha} \partial_\mu\ket{\psi} + i (\partial_\mu\alpha) e^{i\alpha} \ket{\psi}
\end{equation}
and the inner products between the state vector and its derivative transform as 
\begin{equation}
	\ev{\partial_\mu}{\psi}
	\mapsto \ev{\partial_\mu}{\psi} + i (\partial_\mu\alpha).	
\end{equation}
To avoid the cumbersomeness caused by these additional terms, we use the covariant derivative \(\ket{\nabla_\mu\psi} := (\mathbf{1} - \op{\psi}) \partial_\mu \ket{\psi}\) instead of the ordinary derivatives for convenience. 
The covariant derivatives are always orthogonal to the state vectors and changes as \(\ket{\nabla_\mu\psi}\mapsto e^{i\alpha}\ket{\nabla_\mu\psi}\) under the gauge transformation \(\ket\psi\mapsto e^{i\alpha}\ket\psi\).
Moreover, it satisfies that 
\(\partial_\mu \op\psi = \op{\nabla_\mu\psi}{\psi} + \op{\psi}{\nabla_\mu\psi}\)
and 
\(\partial_\mu \ip\psi = \ip{\nabla_\mu\psi}{\psi} + \ip{\psi}{\nabla_\mu\psi}\).
The QFI matrix is totally determined by the mathematical structure of the state vectors, their covariant derivatives, and the prior probability of each state, see Fig.~\ref{fig:overlap} for an illustration.

\begin{figure}
	\includegraphics[]{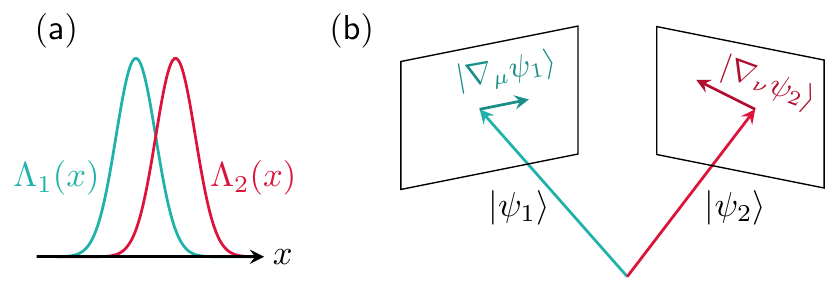}
	\caption{\label{fig:overlap}
		Classical and quantum incoherent-source models.
		(a) In classical model, the state from each source is a probability distribution.
		(b) In quantum model, the (pure) state from each source is represented by a state vector \(\ket{\psi_j}\) in a Hilbert space.
		The infinitesimal transformation of a state vector can be characterized by its covariant derivatives \(\ket{\nabla_\mu\psi_j}\) in the tangent plane.
		The state vectors, their covariant derivatives, and the prior probabilities of each state  determine the QFI matrix.
	}
\end{figure}

\section{Curse on the QFI}

Now, we say that a set \(\{ \ket{\psi_j} \}\) of state vectors satisfies \emph{the local orthogonalization condition} if 
\begin{equation} \label{eq:ortho_condition}
	\ip{\psi_j}{\psi_k} = \ip{\psi_j}{\nabla_\mu\psi_k} = 0
	\quad \forall\ j \neq k 
	\qand \forall\ \mu.
\end{equation}
In such case, it can be shown that \(L_\mu = 2 \sum_j \partial_\mu \op{\psi_j}\) is the SLD operator for \(\rho\) with respect to the parameter \(\theta_\mu\) and thus the entries of the QFI matrix are simply given by 
\begin{equation}\label{eq:preQFI}
	\widetilde\qfi_{\mu\nu} := 4 \sum_{j=1}^n w_j \Re\ip{\nabla_\mu\psi_j}{\nabla_\nu\psi_j}.	
\end{equation}
This QFI matrix can be understood as the average of the QFI matrix in each state vectors with the prior probability.

\sect{Curse matrix}
We shall focus on the cases where the local orthogonalization condition Eq.~\eqref{eq:ortho_condition} is not satisfied. 
The QFI matrix can be formally written as 
\begin{equation} \label{eq:cursed}
	\qfi = \widetilde\qfi - \curse,
\end{equation}
where \(\widetilde\qfi\) is the QFI matrix given by Eq.~\eqref{eq:preQFI} and \(\curse\) is a matrix manifesting the influence of the overlap between the state vectors on the QFI matrix. 
We call \(\curse\) the curse matrix, as it inclines to diminish the QFI when the local orthogonalization condition does not hold.

Assume that the state vectors \(\ket{\psi_j}\) and their non-vanishing derivatives \(\ket{\nabla_\mu\psi_j}\) are all linearly-independent and the prior weights \(w_j\) are all independent of \(\theta\). 
We show in the Appendix~\ref{app:proof} that the curse matrix for the QFI can be expressed as
\begin{align} \label{eq:curse_term}
	\curse_{\mu\nu} & = \Re \tr(\sfX_\mu \sfA \sfX_\nu),
\end{align}
where \(\sfA\) is the Gram matrix of \(\{\ket{\phi_j}\}\) with \(\ket{\phi_j} := \sqrt{w_j} \ket{\psi_j}\) being defined for convenience, i.e., \(\sfA_{jk}:=\ip{\phi_j}{\phi_k}\), and \(\sfX_\mu\) are \(n\times n\) Hermitian matrices determined by 
\begin{equation} \label{eq:sldlike}
	\frac12 (\sfA \sfX_\mu + \sfX_\mu \sfA) = \sfZ_\mu
\end{equation}
with \(\sfZ_\mu\) being the \(n\times n\) Hermitian matrix given by 
\begin{equation}\label{eq:Zmu}
	[\sfZ_\mu]_{jk} = i \ip{\phi_j}{\nabla_\mu \phi_k} - i\ip{\nabla_\mu \phi_j}{\phi_k}.
\end{equation}
Note that all diagonal entries of \(\sfZ_\mu\) are zero according to the above definition, for any state vector is orthogonal to its covariant derivative.
Due to Eq.~\eqref{eq:sldlike}, the entries of the curse matrix can also be written as
\begin{equation} \label{eq:curse_simple}
	\curse_{\mu\nu}=\tr(\sfX_\mu\sfZ_\nu), 	
\end{equation}
 which may be more convenient for practical calculation.

To understand and even dispel Rayleigh's curse on parameter estimation with incoherent sources, we need to further analyze the curse matrix.
Let \(v\) be a column vector in \(\mathbb R^d\) and represent a direction in the parameter space.
We interpret the quantity \(v^\top \curse v\) as the QFI curse along \(v\), where \(\top\) denotes matrix transpose. 
The meaning of the QFI curse along \(v\) can be understood in an intuitive way by resorting to the second form of the QCRB~\cite[Chapter VIII]{Helstrom1976}:
\(v^\top \mathcal E^{-1} v \leq v^\top \qfi v\) 
for any \(v\) in \(\mathbb R^d\).
Through a contour plot of \(v^\top \mathcal E^{-1} v\) as a function of \(v\), the second form of QCRB implies that, for each positive number \(c\), the concentration ellipsoid~\cite{Friendly2013} given by those \(v\) that are restricted by \(v^\top \mathcal E^{-1} v = c\) always lies outside the ellipsoid whose equation is \(v^\top \qfi v = c\).
See Fig.~\ref{fig:ellipsoid} for an illustration.

\begin{figure}[tb]
	\centering
	\includegraphics[]{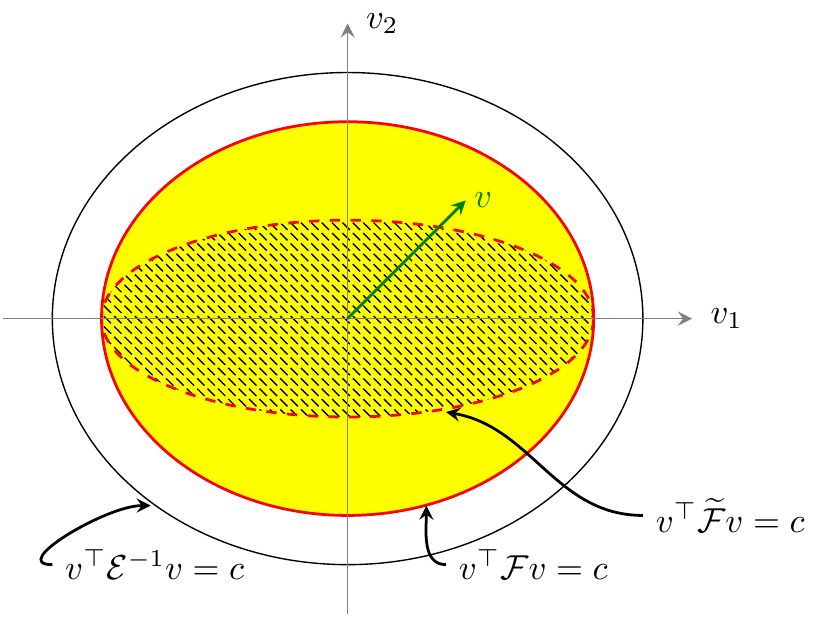}
	\caption{
		Geometric view of the QFI curse along a direction \(v\).
		Here, we consider two-parameter estimation as an example.
		The contour lines of \(v^\top \mathcal E^{-1} v\), \(v^\top \mathcal F v\), and \(v^\top \widetilde{\mathcal F} v\) are all ellipses. 
		For each positive number \(c\), the ellipse given by \(v^\top \mathcal E^{-1} v = c\) must lay outside that given by \(v^\top \mathcal F v=c\), which in turn must lay outside that given by \(v^\top \widetilde{\mathcal F} v=c\).
		The distance between the \(v^\top \widetilde{\mathcal F} v\) and \(v^\top {\mathcal F} v\) is caused by the curse on QFI along \(v\). 
	}
	\label{fig:ellipsoid}
\end{figure}

Note that \( v^\top \curse v = \trace(\sfX \sfA \sfX) \geq 0\) is satisfied for any \(v\), where we have defined \(\sfX:=\sum_\mu v_\mu \sfX_\mu\).
In other words, the curse matrix \(\curse\) is positive semidefinite.
This reflects the negative influence of the overlaps between state vectors on the QFI of the mixed states.
In fact, the positivity of the curse matrix can be understood from the perspective of the monotonicity of QFI under quantum operations~\cite{Petz1996,Hiai2014,Lu2015} as follows.
Let \(\tilde\rho = \sum_j w_j \op{\psi_j} \otimes \op{j}\) with \(\ket{j}\) being a set of orthonormal states in an ancillary Hilbert space and independent of \(\theta\).
It can be shown that the QFI matrix for \(\tilde\rho\) is the same of \(\tilde\qfi\) given by Eq.~\eqref{eq:preQFI}.
The incoherent-source state \(\rho\) can be obtained by performing on \(\tilde\rho\) the partial trace operation with respective to the ancilla, which is a quantum operation and thus cannot increase the QFI according to the monotonicity of QFI.
Therefore, the curse matrix \(\curse\) must be positive semidefinite.

\section{Immunity to the QFI curse}

When \(v^\top \curse v\) vanishes, we say that the QFI is \emph{immune} to Rayleigh's curse along the direction \(v\).
Moreover, we say that the immunity is \emph{global} if it holds for a direction \(v\) that is independent of the true value of the parameters of interest;
Otherwise, we say the immunity is \emph{local}.
Finding all vectors \(v\) such that \(v^\top \curse v\) vanishes is mathematically equivalent to seeking the kernel of the curse matrix.
We give the necessary and sufficient condition on the immunity to Rayleigh's curse along a direction in the following theorem.

\begin{thm}\label{thm:ns}
	The QFI is immune to Rayleigh's curse along a parameter direction \(v\), if and only if 
	\begin{equation} \label{eq:ns_condition}
		\ip{\psi_j}{D_v \psi_k} = \ip{D_v\psi_j}{\psi_k}  \quad\forall\, j<k
	\end{equation}
	are satisfied, where we \(\ket{D_v\psi} := \sum_\mu v_\mu \ket{\nabla_\mu\psi}\) is the directional derivative of a state vector \(\ket\psi\) along \(v\).
\end{thm}

\begin{proof}
Let us define \(\sfZ = \sum_\mu v_\mu \sfZ_\mu\), where \(\sfZ_\mu\) is given by Eq.~\eqref{eq:Zmu}.
The condition Eq.~\eqref{eq:ns_condition} is equivalent to \(\sfZ = \sfO\), where \(\sfO\) denotes the zero matrix.
Notice that \(\sfZ = (\sfA \sfX + \sfX \sfA) / 2\) according to their definitions.
The sufficiency of Eq.~\eqref{eq:ns_condition} is then evident by observing that \(v^\top \curse v = \trace(\sfX \sfZ)\).
To prove the necessity of condition Eq.~\eqref{eq:ns_condition}, note that \(v^\top \curse v = \tr(\sfX \sfA \sfX)\) vanishes if and only if \(\sqrt{\sfA}\sfX\) is the zero matrix, for \(\tr(\sfX \sfA \sfX)\) is the Hilbert-Schmidt norm of \(\sqrt\sfA \sfX\).
In such case, \(\sfZ\) must be a zero matrix due to \(\sfZ = (\sfA \sfX + \sfX \sfA) / 2\).	
\end{proof}

\sect{Formula for the curse matrix}
To calculate the curse matrix, one needs to solve \(\sfX_\mu\) from Eq.~\eqref{eq:sldlike}.
This can be done in general by invoking the eigenvalue decomposition of \(\sfA\).
Let \(\sfP_q\) be the eigen-projection of \(\sfA\) with the eigenvalue \(\lambda_q\).
It then follows from Eq.~\eqref{eq:sldlike} that 
\begin{equation}
	\sfP_q \sfZ_\mu \sfP_{q'}  =\frac{\lambda_q + \lambda_{q'}}{2} \sfP_q \sfX_\mu \sfP_{q'},
\end{equation}
where we have used \(\sfA \sfP_{q'} = \lambda_{q'} \sfP_{q'}\) and \(\sfP_q \sfA = \lambda_q \sfP_q\).
Therefore, we get
\begin{equation}\label{eq:computable_curse}
	\sfX_\mu = \sum_{qq'} \sfP_q \sfX_\mu \sfP_{q'} 
	= \sum_{qq'}\frac{2}{\lambda_q + \lambda_{q'}} \sfP_q \sfZ_\mu \sfP_{q'}.
\end{equation}
As a result, the entries of the curse matrix can be given by 
\begin{equation} \label{seq:curse}
	\curse_{\mu\nu} = \tr(\sfX_\mu \sfZ_\nu) 
	= \sum_{qq'} \frac{2}{\lambda_q + \lambda_{q'}} \tr(\sfZ_\mu \sfP_q \sfZ_\nu \sfP_{q'}).
\end{equation}

For two-source problems, the QFI curse along \(v\) can be simply expressed as (see Appendix~\ref{app:twosource} for a detailed derivation) 
\begin{align}\label{eq:twosource_f}
	v^\top \curse v 
	&= 4 w_1 w_2 \qty[\abs{\gamma}^2 + \frac{(\Re\gamma\ip{\psi_2}{\psi_1})^2}{1 - \abs{\ip{\psi_2}{\psi_1}}^2}]
\end{align}
with \(\gamma := i \ip{\psi_1}{D_v \psi_2} - i \ip{D_v \psi_1}{\psi_2}\) being defined for brevity.
Without loss of generality, assume that \(\ip{\psi_2}{\psi_1}\) is real and nonnegative, which can always be implemented by changing the phases of \(\ket{\psi_1}\) and \(\ket{\psi_2}\), i.e., fixing a specific gauge.
Then, \(\ip{\psi_2}{\psi_1}\) equals the fidelity between the two state vectors and will be denoted by \(f\).
For two close sources such that the fidelity \(f\approx1\),  we have \(\gamma\approx0\), for \(f \to 1\) implies \(\ket{\psi_1} \to \ket{\psi_2}\).
To see the limit of Eq.~\eqref{eq:twosource_f} as \(f \to 1\), we introduce a normalized vector \(\ket\tau := (\ket{\psi_1} - \ket{\psi_2}) / \sqrt{2 - 2 f}\) so that \(\gamma\) can be written as
\begin{equation}\label{eq:gammatau}
	\gamma = i \sqrt{2-2f} \qty(\ip{\tau}{D_v \psi_2} + \ip{D_v \psi_1}{\tau}).
\end{equation}
Notice that \(\ket\tau\) is the unit vector representing the direction from the tip of \(\ket{\psi_2}\) to that of \(\ket{\psi_1}\) and thus is well defined under the limit \(\ket{\psi_1} \to \ket{\psi_2}\). 
Substituting Eq.~\eqref{eq:gammatau} into Eq.~\eqref{eq:twosource_f}, we get
\begin{align}
	v^\top \curse v &=  8 w_1 w_2 \qty[ 
		(1-f) \abs{\tilde\gamma}^2 + \frac{f^2 (\Re\tilde\gamma)^2}{1 + f}
	],
\end{align}
where \(\tilde\gamma := i\ip{\tau}{D_v\psi_2} + i\ip{D_v\psi_1}{\tau}\) is in general a complex number.
For two close sources such that the fidelity \(f\) is close to unit, we get
\begin{align}
	v^\top \curse v 
		\xrightarrow{f \to 1} 
		4 w_1 w_2 (\Re\tilde\gamma)^2.
\end{align}
This tell us that the local immunity to the QFI curse along a direction \(v\) can be achieved for two sources, if the quantity \(\Re\tilde\gamma\) goes to zero as the two state vectors become close.

\section{Unitary parameter estimation}

We here consider the cases where the state vectors are given by \(\ket{\psi_j} = U_j \ket{\psi}\), where \(U_j\) is a family of unitary operators depending on the parameter \(\theta\) of interest and \(\ket\psi\) is the initial state in common.
The covariant derivatives of the state vectors are given by
\begin{equation}
	\ket{\nabla_\mu\psi_j} 
	= \qty(\mathbf1 - \op{\psi_j}) \qty(\partial_\mu U_j) \ket\psi
	= - i G_{\mu,j} \ket{\psi_j},
\end{equation}
where we have defined 
\begin{equation}
	G_{\mu,j} := i(\partial_\mu U_j) U_j^\dagger - i \ev{(\partial_\mu U_j)U_j^\dagger }{\psi_j} \mathbf1
\end{equation}
with \(\mathbf1\) being the identity operator.
Note that the operator \(G_{\mu,j}\) is Hermitian due to \(\partial_\mu (U_j U_j^\dagger) = (\partial_\mu U_j) U_j^\dagger + U_j (\partial_\mu U_j)^\dagger = 0\) and satisfies \(\ev{G_{\mu,j}}{\psi_j} = 0\).
Therefore, we get 
\begin{align}
	 [\mathsf Z_\mu]_{jk} &= \mel{\psi_j}{\qty(G_{\mu,k} + G_{\mu,j}) }{\psi_k}.
\end{align}
The necessary and sufficient condition on the immunity to Rayleigh's curse along a direction \(v\), i.e., Eq.~\eqref{eq:ns_condition}, becomes
\begin{equation}\label{seq:immunity_unitary}
	\sum_\mu v_\mu \mel{\psi_j}{\qty(G_{\mu,k} + G_{\mu,j}) }{\psi_k} = 0
	\quad \forall\ j<k.
\end{equation}

We further assume that each state vector \(\ket{\psi_j}\) depends on an individual parameter \(\theta_j\) standing for the shifted location from a common initial state \(\ket{\psi}\) by a fixed unitary transformation, i.e., \(\ket{\psi_j} = \exp(-i\theta_j H) \ket{\psi}\) with \(H\) being an Hermitian operator and independent of \(\theta\).
This class of problems includes the superresolution for two incoherent optical point sources~\cite{Tsang2016b}, where the state \(\ket{\psi_j}\) is for a photon's spatial degree of freedom at the image plane from the \(j\)-th source. 
To be specific, assuming a spatially-invariant unit-magnification imaging system, the state vectors from each source are in the form of \(\ket{\psi_j} = e^{-i\theta_j H}\ket\psi\), where \(H\) is the momentum operator (\(\hbar\) is set to be unit) and \(\ket\psi\) is the state of a photon's spatial degree of freedom when the point optical source is located at the origin of coordinate.
In Ref.~\cite{Tsang2016b}, Tsang {\it et al.} showed that the QFI about the separation, i.e., \(\theta_1-\theta_2\), is independent of its true value, provided that the two sources are equally weighted (\(w_1 = w_2\)) and the normalized amplitude-point-spread function for the imaging system has a constant position-independent phase so that \(\ip{\psi_2}{\psi_1}\), \(\ip{\psi_2}{\nabla_1\psi_1}\), and \(\ip{\nabla_2\psi_2}{\psi_1}\) have the same phase.
In our context, it can be reinterpreted that the QFI has the global immunity to Rayleigh’s curse along \(v = (1 / \sqrt2, -1 / \sqrt2)\), for the directional derivative along \(v\) is equivalent to the partial derivative with respect to the separation \((\theta_1 - \theta_2) / \sqrt2\) when we use the rotated parameters \((\theta_1 + \theta_2) / \sqrt2\) and \((\theta_1 - \theta_2) / \sqrt2\) as the parameters of interest.
We will explain it in detail and generalize the results in what follows.

Without loss of generality, we assume that \(\ev{H}{\psi}=0\) henceforth; 
Otherwise, we can take the gauge transformation \(\ket{\psi_j} \mapsto \exp(i \theta_j \ev{H}{\psi}) \ket{\psi_j}\) and replace \(H\) by \(\Delta H := H - \ev{H}{\psi} \mathbf 1\).
In such case, we have \(\ket{\nabla_\mu\psi_j} = -i \delta_{\mu j} H \ket{\psi_j}\).
Then, The following corollary straightforwardly implied by Theorem~\ref{thm:ns}.

\begin{cor}
	For \(\rho = \sum_j w_j e^{-i\theta_j H} \op\psi e^{i\theta_j H}\), the QFI is immune to Rayleigh's curse along \(v\) if and only if
	\begin{equation}\label{eq:ns2}
		(v_j + v_k) \ev{H e^{i (\theta_j -\theta_k) H}}{\psi} = 0 \quad \forall\ j<k,
	\end{equation}	
	where \(\ev{H}{\psi}=0\) is assumed.
\end{cor}
Since the quantity \(\ev{H e^{i (\theta_j -\theta_k) H}}{\psi}\) in general depends on the true values of \(\theta\), the global immunity can only be achieved when \(v_j = -v_k\) for all different \(j\) and \(k\), which is only possible for the case of two sources.
For two-source cases, this global immunity along \(v\propto(1,-1)\) is irrelevant to the prior weights \(w_j\) and the specific form of \(H\) and \(\ket{\psi}\), implying that such a global immunity still holds for two unequal weighted sources and for point-spread functions with position-dependent phases.

The curse matrix for the model considered above can be obtained as follows.
By defining a characteristic function as \(g(x) := \ev{e^{-i x H}}{\psi}\) and noting that \(g'(x) = -i \ev{H e^{-i x H}}{\psi}\), it follows from Eq.~\eqref{eq:twosource_f} that 
\begin{align}
	v^\top \curse v 
	&= 4 w_1 w_2 (v_1 + v_2)^2 \Theta(\theta_2 - \theta_1) \nonumber \\
	&\mbox{with }
	\Theta(x) := \abs{g'(x)}^2 + \frac{[\Im g'(x) g(x)]^2}{1 - \abs{g(x)}^2}.
\end{align}
This form implies that the curse matrix can be simply written as 
\begin{equation}\label{eq:curse_matrix}
	\curse = 4 w_1 w_2 \Theta(\theta_2 - \theta_1) \mqty[1 & 1 \\ 1 & 1].
\end{equation}
When we rotate the parameters \(\theta_1\) and \(\theta_2\) to \(\theta'_1=(\theta_1+\theta_2)/\sqrt2\) and \(\theta'_2 = (\theta_1-\theta_2)/\sqrt2\), the curse matrix is diagonal as \(\curse = 8 w_1 w_2 \Theta(\theta_2 - \theta_1) \mathrm{diag}\{2,0\}\).
It is worth noting that for the case of locating two incoherent optical point sources, the QFI matrix regarding \(\theta'_1\) and \(\theta'_2\) is no longer diagonal when the two intensities are unequal~\cite{Rehacek2017}.
However, here, it can be seen that the curse matrix is still diagonal regarding \(\theta'_1\) and \(\theta'_2\).
We give a concrete example to illustrate this interesting phenomena.
Assume that the state vector \(\ket\psi\) is given by 
\begin{equation}
	\ket\psi = \int_{-\infty}^\infty 
	\qty(2\pi\sigma^2)^{-1/4} \exp[-\frac{x^2}{4\sigma^2}] \ket x \dd x,
\end{equation} 
where \(\sigma\) is a characteristic length, \(\ket x\) denotes the eigen-ket of the position operator, and the Hamiltonian is the momentum operator, i.e., \(H = - i\, \partial / \partial x\).
In such case, we have 
\begin{equation}\label{eq:example}
	\qfi = \frac1{\sigma^2}\mqty[w_1 & 0 \\ 0 & w_2] 
	+ \frac{w_1 w_2 s^2}{4 \sigma ^4} 
		\exp(-\frac{s^2}{4 \sigma ^2})
		\mqty[1 & 1 \\ 1 & 1].
\end{equation}
with \(s = \theta_2 - \theta_1\).
The first term in the right hand side of Eq.~\eqref{eq:example} is \(\tilde\qfi\) and the second one is the curse matrix \(\curse\).
For either \(s=0\) or \(s=\infty\), the curse matrix vanishes so that the QFI matrix is the same as \(\tilde\qfi\).
For an intermediate separation (\(s = 2\sigma\)), we in Fig.~\ref{fig:unequal_intensity} plot the QCRB ellipses given by \(v^\top \qfi v = 1 / \sigma^2\) and compare them with the case of \(s=\infty\).
The global immunity to the QFI curse along the direction \((1,-1)\) is manifested in the fact that the two ellipses in Fig.~\ref{fig:unequal_intensity} always touched along this direction( the gray dotted line), no matter the intensities of the two source are equal or notx.

\begin{figure}[tb]
	\raisebox{7cm}{(a)}
	\includegraphics[]{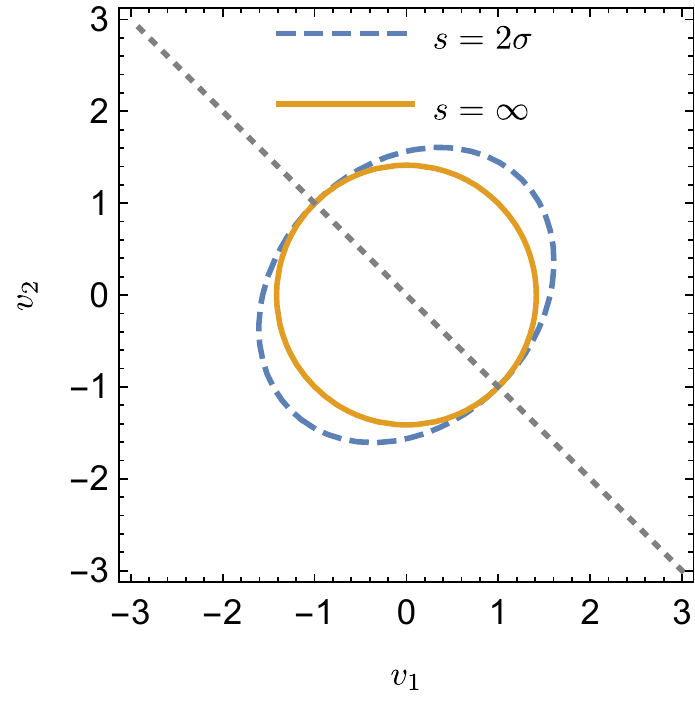}
	\par
	\raisebox{7cm}{(b)}
	\includegraphics[]{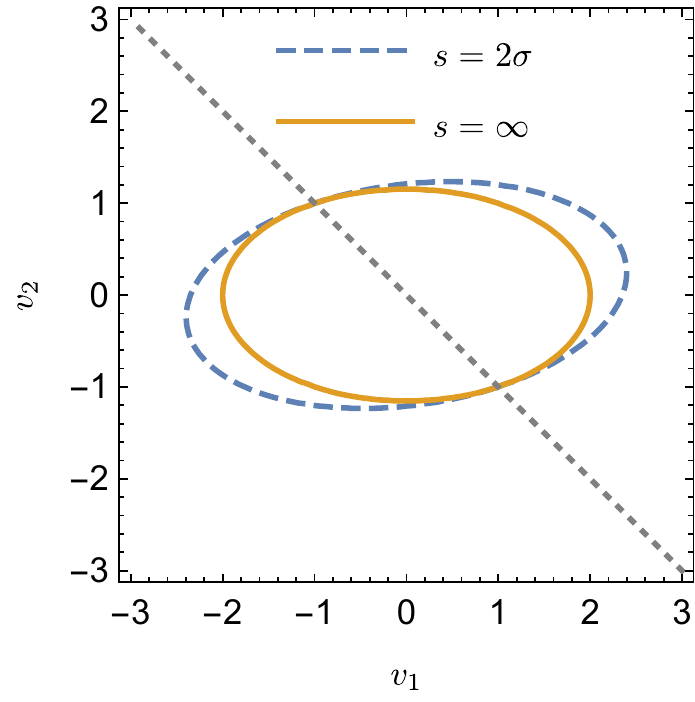}
	\caption{
		QCRB ellipses given by \(v^\top \qfi v = 1 / \sigma^2\).
		(a) We set the unequal intensities as \(w_1 = w_2 = 1/2\).
		(b) We set the unequal intensities as \(w_1 = 1/4\) and \(w_2 = 3/4\).
		The gray dotted line indicates the direction \(v=(1,-1)\).
	}
	\label{fig:unequal_intensity}
\end{figure}

For two close sources such that \(|\theta_2 - \theta_1| \ll 1\), we can expand the relevant quantities with respect to \(\theta_2 - \theta_1\) as follows:
\begin{align}
	g(\delta) &= 1  - \frac{\delta^2}{2} \ev{H^2}{\psi} + \mathcal O(\delta^3), \\
	g'(\delta)
 	&= - \delta \ev{H^2}{\psi} + \frac{i\delta^2}{2} \ev{H^3}{\psi} + \mathcal O(\delta^3),
\end{align}
where \(\delta\) is short for \(\theta_2 - \theta_1\) and we have used \(g'(0) = -i \ev{H}{\psi}=0\) and \(g''(0) = - \ev{H^2}{\psi}\).
Therefore, for small \(\delta\), we get 
\begin{equation} \label{eq:one_d_curse}
	\Theta(\delta) \approx \delta^2 \qty(
		\ev{H^2}{\psi}^2 + \frac{\ev{H^3}{\psi}^2}{4 \ev{H^2}{\psi}}
	).
\end{equation}
As \(\delta \to 0\), we can see that the local immunity of the QFI curse occurs along any direction \(v\).

In short, we have shown that, for the two-incoherent-source parameter estimation model \(\rho=\sum_{j=1,2} w_j \op{\psi_j}\) with \(\ket{\psi_j} = e^{-i\theta_j H} \ket\psi\), the QFI is globally immune along the direction \(v\propto(1,-1)\) and locally immune along all direction, to the Rayleigh's curse at the region of small \(\delta\).



\section{Discussion} 
To summarize, we have generalized the concept of Rayleigh's curse to quantum parameter estimation with incoherent sources. 
To manifest the influence of the overlap between quantum states from each sources on estimating the carried parameters, we have defined the curse matrix and introduced the global and local immunities to the QFI curse accordingly.
Using the technology of non-orthonormal basis, we have derived the computable formula for the curse matrix and also given a compact form of the curse matrix for the case of two incoherent sources.
We have applied the curse matrix to the estimation of one-dimensional location parameters.
We believe that our approach is useful to quantum metrology with incoherent sources.

In a recent work~\cite{Lupo2020a}, Lupo, Huang, and Kok used a fidelity approach to show that quantum limits to incoherent imaging of multiple point sources can be achieved by linear interferometry.
The fidelity approach can be used to obtain an analytic expression of the QFI matrix for the case of estimating the three-dimensional positions of two incoherent point sources in the paraxial regime. 
Yet, it is still not easy to analytically calculate the QFI matrix for more general cases with the fidelity approach.
Our general formula Eq.~\eqref{eq:computable_curse} for the curse matrix and its special form Eq.~\eqref{eq:twosource_f} for the two-source case may be used to investigate the estimation of the parameters that are imprinted in quantum states in more general ways than the unitary parameter.
Besides, Fiderer {\it et al}.~\cite{Fiderer2020} recently derived a general expression of the QFI matrix using the block-vectorization technology together with the non-orthogonal basis technology~\cite{Genoni2019,Napoli2019}, the latter also plays an important role in our present work.
All these new methods~\cite{Lupo2020a,Fiderer2020,Genoni2019,Napoli2019} enrich the toolbox for multiparameter estimation with incoherent sources, by providing approaches for calculating the QFI matrix of low-rank quantum states without resorting to the cumbersome procedure of Schmidt orthogonalization for an orthonormal basis to solve the SLD operator.

In this  work, we mainly focus on the condition on the vanishing of the QFI curse along a parameter direction.
It will be also interesting to investigate the condition on the maximal QFI curse along a parameter direction, which corresponds to the vanishing of QFI along a parameter direction. 
Bisketzi, Branford, and Datta have numerically demonstrated that in the regime of small separations, no more than two independent parameters can be effectively estimated, as the QFI matrix is no more than rank two~\cite{Bisketzi2019}.
Calculating the curse matrix for the same scenario as Ref.~\cite{Bisketzi2019} and investigating the condition on the maximal QFI curses will help to understand the above-mentioned phenomena.

Some important relevant aspects need to be considered for practical problems of multiparameter estimation with incoherent sources.
The curse matrix approach developed in this work is based on the QFI matrix, whose meaning is rooted in the QCRB.
However, the QCRB is not guaranteed to be attainable for multiparameter estimation due to Heisenberg's uncertainty principle.
This is known as the incompatibility problem in quantum multiparameter estimation.
There are some approaches that can be used to investigate the incompatibility problem, e.g., the weak compatibility condition~\cite{Matsumoto2002,Ragy2016}, the bounds on the discrepancy between the Holevo bounds and the QCRB~\cite{Carollo2019a}, and the tradeoff relation between the regrets of Fisher information~\cite{Lu2020b}.
Besides, another important question is how to construct the optimal measurement when the QFI immunity condition is satisfied.

\begin{acknowledgments}
We thank Chandan Datta, Stefano Pirandola, and Animesh Datta for the helpful communications and enlightening discussions. 
This work is supported by 
the National Natural Science Foundation of China (Grants No.~11805048, No.~61871162, and No.~11935012)
and Zhejiang Provincial Natural Science Foundation of China (Grant No.~LY18A050003).
\end{acknowledgments}

\appendix
\section{Derivation of the formula for the curse matrix}
\label{app:proof}
We here give a detailed proof of the formula Eq.~\eqref{eq:curse_term}--\eqref{eq:curse_simple} for the curse matrix.
Our approach is based on the use of non-orthonormal basis of the subspace supporting the underlying density operator and its partial derivatives with respect to the parameters of interest. 

\subsection{Non-orthonormal basis}

We follow Ref.~\cite{Genoni2019} to use non-orthogonal bases for expanding the relevant  operators involved in calculating QFI matrix.
Let \(\mathcal B = \{\ket{\phi_j} \mid j=1,2,\ldots,n\}\) be a set of linearly independent vectors spanning a Hilbert space \(\mathcal{H}_\mathrm{s}\).
Any bounded operator \(A\) acting on \(\mathcal H_\mathrm{s}\) can be uniquely expressed as 
\begin{equation}
	A = \sum_{jk} \sfR(A)_{jk} \op{\phi_j}{\phi_k},
\end{equation}
where \(\sfR(A)_{jk}\) are the coefficients of the expansion.
We call \(\sfR(A)\), the matrix with entries \(\sfR(A)_{jk}\), the \(\sfR\)-matrix for an operator \(A\) in the basis \(\mathcal B\).
In this work, we will use the following properties of the \(\sfR\)-matrices~\citep{Napoli2019,Genoni2019,Soriano2014}: 
\begin{enumerate}
	\item \(\sfR(A^\dagger) = \sfR(A)^\dagger\).
	\item \(\sfR(a A + b B) = a \sfR(A) + b \sfR(B)\) for any bounded linear operators \(A\) and \(B\) on the subspace spanned by \(\mathcal{B}\) and any two complex numbers \(a\) and \(b\).
	\item \(\tr A = \trace[\sfR(A) \Omega]\), where \(\Omega\) is the Gram matrix of \(\mathcal{B}\) and given by \(\Omega_{jk} := \ip{\phi_j}{\phi_k}\).
	\item \(\sfR(AB) = \sfR(A)\Omega\sfR(B)\).
\end{enumerate}
The first and second properties are obvious. 
The third property can be proved as
\begin{align}
	\tr(A) & =\sum_{jk} \sfR(A)_{jk} \tr(\op{\phi_j}{\phi_k})  \nonumber\\
	&=\sum_{jk}\sfR(A)_{jk}\ip{\phi_k}{\phi_j}
	=\tr[\sfR(A)\Omega].
\end{align}
The fourth property can be proved by noting that 
\begin{align}
	AB & =\sum_{jk} \sfR(A)_{jk} \op{\phi_j}{\phi_k}  \sum_{\ell m} \sfR(B)_{\ell m} \op{\phi_\ell}{\phi_m} \nonumber\\
 	&= \sum_{jm} \left[
 		\sum_{k\ell} \sfR(A)_{jk} \ip{\phi_k}{\phi_\ell} \sfR(B)_{\ell m}
 	\right] \op{\phi_j}{\phi_m}.
\end{align}


\subsection{General expression for the curse matrix}

To calculate the QFI of the incoherent-source density operator \(\rho=\sum_j w_j \op{\psi_j}\), Genoni and Tufarelli in Ref.~\cite{Genoni2019} used the non-orthogonal bases constituted by the state vectors \(\ket{\psi_j}\) and their derivatives \(\partial_\mu\ket{\psi_j}\) with respect to \(\theta_\mu\).
In this work, we improve their approach by using the weighted state vectors \(\sqrt{w_j}\ket{\psi_j}\) and their covariant derivatives \(\sqrt{w_j}\ket{\nabla_\mu \psi_j}\) instead to constitute the non-orthogonal basis.
We make the following assumptions on the incoherent-source density operators:
(i) \(w_j\) are independent of \(\theta\) and
(ii) all non-vanishing \(\ket{\nabla_\mu \psi_j}\) and \(\ket{\psi_j}\) are linearly independent.
Let us choose \(\ket{\phi_j} = \sqrt{w_j}\ket{\psi_j}\) for \(j=1,\,2,\,\ldots,\,n\) and take from the non-vanishing covariant derivatives as \(\ket{\phi_j}\) for \(j>n\).
It follows that \(\rho = \sum_{j=1}^n \op{\phi_j}\) and \(\partial_\mu\rho = \sum_{j=1}^n \op{\nabla_\mu \phi_j}{\phi_j} + \op{\phi_j}{\nabla_\mu \phi_j}\).
By partitioning the basis into \(\{\ket{\phi_j} \mid j=1,2,\ldots,n\}\) and the remainder, the \(\sfR\)-matrices for \(\rho\) and \(\partial_\mu\rho\) can be expressed in the  block form 
\begin{equation} \label{eq:block}
	\sfR(\rho) = 
		\mqty[\sfI & \sfO \\
		      \sfO & \sfO]
	\qand
	\sfR(\partial_\mu\rho) = 
		\mqty[\sfO & \sfD_\mu^\dagger \\
		      \sfD_\mu & \sfO],
\end{equation}
respectively, where \(\sfI\) is the identity matrix of size \(n\).

To calculate the QFI matrix, one usually need solve the equation \( (L_\mu \rho + \rho L_\mu) / 2 = \partial_\mu \rho\) for the SLD operator \(L_\mu\). 
We here give a method to calculate the QFI without explicitly solving \(L_\mu\).
One of the advantages of absorbing the weighting factors \(\sqrt{w_j}\) into the basis vectors is that \(\sfR(\rho)\) is idempotent, namely, \(\sfR(\rho)^2 = \sfR(\rho)\).
Using this property and the properties of the \(\sfR\)-matrix, we get
\begin{align}
	\mathcal Q_{\mu\nu} &:= \tr(L_\mu L_\nu \rho) \nonumber \\
	&= \tr[\sfR(L_\mu) \Omega \sfR(L_\nu) \Omega \sfR(\rho) \Omega] \nonumber \\
	&= \tr[\sfR(\rho) \Omega \sfR(L_\mu) \Omega \sfR(L_\nu) \Omega \sfR(\rho)] \nonumber \\
	&= \tr(\Psi_\mu^\dagger \Omega \Psi_\nu), 
\end{align}
where \(\Psi_\mu := \sfR(L_\mu \rho) = \sfR(L_\mu) \Omega \sfR(\rho)\) is defined.
Notice that the \(\mathcal Q_{\mu\nu}\) is known as the geometric tensor.
The real part of \(\mathcal Q\) the QFI matrix.
So knowing the matrices \(\Psi_\mu\) and \(\Omega\) is sufficient to get the QFI matrix. 
The imaginary part of \(\mathcal Q\) is the mean Uhlmman curvature up to a constant~\cite{Uhlmann2011,Carollo2018}, which is strongly related to the attainability of the QCRB for multiparameter estimation~\cite{Matsumoto2002,Ragy2016,Carollo2019a,Lu2020b}.

To proceed, we partition other relevant operators and the Gram matrix into the same block form as Eq.~\eqref{eq:block}:
\begin{align} \label{eq:block2}
	\Psi_\mu = \mqty[
		i \sfX_\mu & \sfO \\
		  \sfY_\mu & \sfO
	],\ 
	\sfR(L_\mu) = \mqty[
		\sfM_\mu & \sfN_\mu^\dagger\\
		\sfN_\mu & \sfK_\mu
	],\ 
	\Omega & = \mqty[
		\sfA & \sfB^\dagger \\
		\sfB & \sfC	
	],
\end{align}
where the imaginary unit before \(\sfX_\mu\) is introduced for later convenience, the all-zero blocks in \(\Psi_\mu\) is a result of \(\Psi_\mu = \sfR(L_\mu) \Omega \sfR(\rho)\) and the special form of \(\sfR(\rho)\) in Eq.~\eqref{eq:block}. 
Note that the square matrices \(\sfA\), \(\sfC\), \(\sfM_\mu\), and \(\sfK_\mu\) are all Hermitian.
With the \(\sfR\)-matrices, the SLD equation \( (L_\mu \rho + \rho L_\mu)/2 = \partial_\mu \rho\) is equivalent to
\((\Psi_\mu + \Psi_\mu^\dagger) / 2 = \sfR(\partial_\mu\rho)\).
Substituting the block form of \(\Psi_\mu\) in Eq.~\eqref{eq:block2} into this SLD equation and comparing block by block both sides, we get \(\sfX_\mu = \sfX_\mu^\dagger\) and \(\sfY_\mu = 2 \sfD_\mu\). 
Substituting the block forms of \(\Omega\) and \(\Psi_\mu\) into \(\mathcal Q_{\mu\nu} = \trace(\Psi_\mu^\dagger \Omega \Psi_\nu)\), we get
\begin{align} \label{eq:Q_block}
	\mathcal Q_{\mu\nu} &= \tr(
		{\mqty[
			-i \sfX_\mu  & 2\sfD_\mu^\dagger \\
			\sfO & \sfO
		]
		\mqty[
			\sfA & \sfB^{\dagger}\\
			\sfB & \sfC
		]
		\mqty[
			i \sfX_\nu & \sfO\\
			2 \sfD_\nu & \sfO
		]}) \nonumber \\	
 	& = \tr(
 			\sfX_\mu \sfA \sfX_\nu
 			+ 4 \sfD_\mu^\dagger \sfC \sfD_\nu
 			+ 2 i \sfD_\mu^\dagger \sfB \sfX_\nu
 			- 2 i \sfX_\mu \sfB^\dagger \sfD_\nu
 		). 
\end{align}
It follows from \(\qfi_{\mu\nu} = \Re \mathcal Q_{\mu\nu}\) that
\begin{align} \label{eq:qfi_block}
	 	\qfi_{\mu\nu} &= \Re\tr(
 			\sfX_\mu  \sfA \sfX_\nu
 			+ 4 \sfD_\mu^\dagger \sfC \sfD_\nu 			
 		)
 		- \tr(\sfZ_\mu \sfX_\nu + \sfX_\mu \sfZ_\nu), 
\end{align}
where we have defined 
\begin{equation}\label{eq:zmu_app}
	\sfZ_\mu := i (\sfB^\dagger \sfD_\mu - \sfD_\mu^\dagger \sfB). 
\end{equation}

To obtain the QFI matrix, we still need to solve the Hermitian matrix \(\sfX_\mu\).
We shall show in the following that \(\sfX_\mu\) is determined solely by \(\sfA\) and \(\sfZ_\mu\) via Eq.~\eqref{eq:sldlike}.
Substituting Eq.~\eqref{eq:block2} into \(\Psi_\mu = \sfR(L_\mu) \Omega \sfR(\rho)\) and comparing both sides block by block, we get
\begin{align}
	i \sfX_\mu  &= \sfM_\mu \sfA + \sfN_\mu^\dagger \sfB, \label{eq:from_sld1}\\
	2\sfD_\mu  &= \sfN_\mu \sfA + \sfK_\mu \sfB.  \label{eq:from_sld2}
\end{align}
Interestingly, the dependence of \(\sfX_\mu\) on the blocks of the \(\sfR\)-matrix for the SLD operator \(L_\mu\), i.e., \(\mathsf M_\mu\), \(\mathsf N_\mu\), and \(\mathsf K_\mu\), can be eliminated. 
This can be done as follows.
First, by noting that the matrices \(\mathsf M_\mu\), \(\mathsf K_\mu\), and \(\sfA\) are Hermitian, it follows from Eq.~\eqref{eq:from_sld1} and \eqref{eq:from_sld2} that 
\begin{align}
 	\sfA \sfX_\mu + \sfX_\mu \sfA
	&= -i \sfA \sfN_\mu^\dagger \sfB + i \sfB^\dagger \sfN_\mu \sfA, \label{eq:from_sld3} \\
	2i \sfB^\dagger \sfD_\mu - 2i \sfD_\mu^\dagger \sfB 
	&= i \sfB^\dagger \sfN_\mu \sfA - i\sfA \sfN_\mu^\dagger \sfB. \label{eq:from_sld4}
\end{align}
Therefore, we get
\begin{equation}\label{eq:sldlike_sm}
	\frac12(\sfA \sfX_\mu + \sfX_\mu \sfA) 
	= i \sfB^\dagger \sfD_\mu - i \sfD_\mu^\dagger \sfB 
	= \sfZ_\mu.
\end{equation}
This relation will play an important role in simplifying the expression of the QFI matrix.
Concretely, we have
\begin{align}
	\Re\tr(\sfX_\mu \sfA \sfX_\nu) & =\tr(\sfX_\mu \frac{\sfA \sfX_\nu + \sfX_\nu \sfA}{2}) \nonumber\\
	&= \tr(\sfX_\mu\sfZ_\nu)
\end{align}
and also \(\Re\trace(\sfX_\mu \sfA \sfX_\nu) = \trace(\sfX_\nu \sfZ_\mu)\) by interchanging the subscripts \(\mu\) and \(\nu\).
Therefore, Eq.~\eqref{eq:qfi_block} can be simplified to 
\begin{equation} \label{eq:qfi_dcdxax}
	\qfi_{\mu\nu} = 4\Re\tr(\sfD_\mu^\dagger \sfC \sfD_\nu) - \Re\tr(\sfX_\mu \sfA \sfX_\nu).
\end{equation}

The covariant derivatives \(\ket{\nabla_\nu\phi_j}\) have two indexes \(\nu\) and \(j\), so it is convenient to use \(\nu\) and \(j\) as a composite index for the matrices \(\sfB\), \(\sfC\), and \(\sfD_\mu\). 
The entries of \(\sfD_\mu\) can be denoted by \([\sfD_\mu]_{\nu j,k}\).
It follows from 
\begin{equation}
	\partial_\mu \rho = \sum_{j=1}^n \op{\nabla_\mu \phi_j}{\phi_j} + \op{\phi_j}{\nabla_\mu \phi_j}
\end{equation}
that 
\begin{equation}
	[\sfD_\mu]_{\nu j,k} = \delta_{\mu\nu} \delta_{jk}.
\end{equation}
Therefore, we have 
\begin{align}
	\tr(\sfD_\mu^\dagger \sfC\sfD_\nu) 
	&= \sum_{j,k,\ell=1}^n \sum_{\alpha,\beta=1}^d [\sfD_\mu^\dagger]_{j,\alpha k} \sfC_{\alpha k, \beta\ell} [\sfD_\nu]_{\beta \ell,j}  \nonumber\\
	&= \sum_{j,k,\ell=1}^n \sum_{\alpha,\beta=1}^d \delta_{\mu\alpha}\delta_{jk} \sfC_{\alpha k, \beta\ell} \delta_{\beta\nu} \delta_{\ell,j}  \nonumber\\
	&= \sum_{j=1}^n \sfC_{\mu j,\nu j} 
	=\sum_{j=1}^n \ip{\nabla_\mu \phi_j}{\nabla_\nu \phi_j} \label{eq:dcd}
\end{align}
and 
\begin{align}
	[\sfZ_\mu]_{jk}
	& = i \sum_{\alpha,\beta=1}^d [\sfB^\dagger]_{j, \alpha\ell} [\sfD_\mu]_{\alpha\ell, k}
	 - [\sfD_\mu^\dagger]_{j,\alpha\ell} \sfB_{\alpha\ell,k} \nonumber\\
 	& = i (\sfB_{\mu k,j}^*-\sfB_{\mu j,k}) \nonumber\\
 	& =i\ip{\phi_j}{\nabla_\mu\phi_k} - i\ip{\nabla_\mu \phi_j}{\phi_k}. \label{eq:Zmu2}
\end{align}
Notice that the last expression in Eq.~\eqref{eq:dcd} equals to \(\tilde\qfi_{\mu\nu}\) defined by Eq.~\eqref{eq:preQFI}.
So it follows from Eq.~\eqref{eq:qfi_dcdxax} that \(\curse_{\mu\nu} = \Re\tr(\sfX_\mu \sfA \sfX_\nu)\), which is Eq.~\eqref{eq:curse_term}.
The last expression in Eq.~\eqref{eq:Zmu2} is the same as Eq.~\eqref{eq:Zmu}.



\section{Two-source formula for the QFI curse}
\label{app:twosource}
We here give the proof of the two-source formula for the QFI curse, i.e., Eq.~\eqref{eq:twosource_f}.
Let us consider the 2-rank density operators, where the state vectors \(\ket{\psi_j}\) in the mixed states comes from two sources. 
For 2-rank density operators, the matrix \(\sfA\) can always be written as 
\begin{equation}
	\sfA = \frac12(\sfI + \vb{r} \cdot \vb*{\sigma}) = \frac12 \qty(\sfI + \sum_{j=1}^3 r_j \sigma_j),
\end{equation}
where \(\vb{r}=(r_1,r_2,r_3) \in \mathbb{R}^3\) and \(\vb*{\sigma} = (\sigma_1,\sigma_2,\sigma_3)\) is a vector of Pauli matrices. 
The eigenvalues of \(\sfA\) are given by
\begin{equation}
	\lambda_\pm = \frac{1 \pm r}{2} \mbox{ with }
	r := |\vb{r}| = \sqrt{r_1^2 + r_2^2 + r_3^2}.
\end{equation}
The corresponding eigen-projections are given by 
\begin{equation}\label{seq:P}
	\sfP_\pm = \frac12(\sfI \pm \vb{n} \cdot \vb*{\sigma}),
\end{equation}
where \(\vb{n}=\vb{r}/r\) is a unit vector. 
To calculate the curse matrix, we first need to obtain \(\sfZ_\mu\). 
For a given vector \(v\) in the parameter space \(\mathbb R^n\), define \(\sfZ = \sum_\mu v_\mu \sfZ_\mu\).
It then follows from Eq.~\eqref{seq:curse} that the curse along the direction \(v\) can be expressed as
\begin{align}\label{seq:s0}
	v^\top \curse v &= 2 \tr(\sfZ \sfP_+ \sfZ \sfP_- + \sfZ \sfP_- \sfZ \sfP_+) \nonumber\\
	& \quad + \frac1{\lambda_+}  \tr(\sfZ \sfP_+ \sfZ \sfP_+)
	+ \frac1{\lambda_-}  \tr(\sfZ \sfP_- \sfZ \sfP_-).
\end{align}
Since \(\sfZ_\mu\) are Hermitian and their diagonal entries are all equal to zero, we can always express \(\sfZ\) as \(\sfZ = \vb z \cdot \vb*\sigma\), where \(\vb z\) is a 3-dimensional real vector with the components \(z_j = \frac12 \sum_\mu v_\mu \tr(\sfZ_\mu \sigma_j)\). 
Using Eq.~\eqref{seq:P} and the algebraic properties of the Pauli matrices, it is easy to show that
\begin{align}\label{seq:s1}
	2\tr(\sfZ \sfP_+ \sfZ \sfP_- + \mathrm{H.c.})
	= \tr[\sfZ^2 - \sfZ (\vb n \cdot \vb*\sigma) \sfZ (\vb n \cdot \vb*\sigma)].
\end{align}
Due to the identities
\(	(\vb a \cdot \vb*\sigma ) (\vb b \cdot \vb*\sigma)  
	=(\vb a \cdot \vb b)\sfI + i(\vb a\times\vb b) \cdot \vb*\sigma \) and 
\( (\vb a \times \vb b) \cdot (\vb c \times \vb d) = (\vb a \cdot \vb c) (\vb b \cdot \vb d) - (\vb b \cdot \vb c) (\vb a \cdot\vb d)\) for any vectors \(\vb a,\vb b,\vb c, \vb d \) in \(\mathbb R^3\),
we get 
\begin{align}\label{seq:s2}
	\tr[\sfZ (\vb n \cdot \vb*\sigma) \sfZ (\vb n \cdot \vb*\sigma)] 
	&= \tr[((\vb z \cdot \vb n)\sfI + i(\vb z\times\vb n) \cdot \vb*\sigma)^2] \nonumber\\
	&= 2 (\vb z \cdot \vb n)^2 - 2 (\vb z \times \vb n) \cdot (\vb z \times \vb n) \nonumber\\
	&= 4 (\vb z \cdot \vb n)^2 - 2 (\vb z \cdot \vb z).
\end{align}
Since \(\sfP_\pm\) are one-dimensional projection operators, it is easy to get 
\begin{align}
	  \sfP_\pm \sfZ \sfP_\pm 
	= \tr(\sfZ \sfP_\pm) \sfP_\pm
	= \pm(\vb n \cdot \vb z) \sfP_\pm
\end{align}
and thus 
\begin{align}\label{seq:s3}
	 \tr[\sfZ \sfP_\pm \sfZ \sfP_\pm] = (\vb n \cdot \vb z)^2.
\end{align}
Putting Eqs.~\eqref{seq:s0}, \eqref{seq:s1}, \eqref{seq:s2} and \eqref{seq:s3} together, we get 
\begin{align}
	v^\top \curse v &= 4\vb z \cdot \vb z  
	+ \qty(\frac1{\lambda_+\lambda_-} - 4) (\vb n \cdot\vb z)^2 \nonumber\\
	&= 4\vb z \cdot \vb z  
	+ \frac{4}{1-r^2} (\vb r \cdot \vb z)^2. \label{eq:s_two_source}
\end{align}

Comparing the definition \(\sfA_{jk} = \sqrt{w_j w_k} \ip{\psi_j}{\psi_k}\) with 
\begin{equation}
	\sfA = \frac12(\sfI + \vb r \cdot \vb*\sigma) = \frac12 \mqty[
		1 + r_3 & r_1 - i r_2 \\
		r_1 + i r_2 & 1 - r_3
	],
\end{equation}
the coefficients of the Pauli vector can be expressed as 
\(r_1 = 2\sqrt{w_1 w_2} \Re\ip{\psi_2}{\psi_1}\), 
\(r_2 = 2\sqrt{w_1 w_2} \Im\ip{\psi_2}{\psi_1}\),
and \(r_3 = w_1 - w_2\). 
Therefore, we get
\begin{align}
	1-r^2 &= 1 - 4 w_1 w_2 \abs{\ip{\psi_2}{\psi_1}}^2 - (w_1 - w_2)^2 \nonumber\\
	 &= (w_1+w_2)^2 - 4 w_1 w_2 \abs{\ip{\psi_2}{\psi_1}}^2 - (w_1 - w_2)^2  \nonumber \\
	 &= 4 w_1 w_2 \qty(1 - \abs{\ip{\psi_2}{\psi_1}}^2), \label{eq:sl1}
\end{align}
where we have used \(w_1 + w_2 = 1\) in the second equality.
Note that \(\vb z_1 - i \vb z_2 = \sfZ_{12} = \sqrt{w_1 w_2} \gamma\) with
\begin{equation}
	\gamma := i \ip{\psi_1}{D_v \psi_2} - i \ip{D_v \psi_1}{\psi_2}
\end{equation}
being defined.
So we get \(\vb z_1 = \sqrt{w_1 w_2} \Re\gamma\) and \(\vb z_2 = - \sqrt{w_1 w_2} \Im\gamma\).
Thus, we have
\begin{align}
	\vb z \cdot \vb z 
	&= w_1 w_2 |\gamma|^2, \label{eq:sl2}\\
	\vb r \cdot \vb z
	&= 2 w_1 w_2 \qty(\Re\ip{\psi_2}{\psi_1} \Re\gamma - \Im\ip{\psi_2}{\psi_1} \Im\gamma) \nonumber\\
	&= 2 w_1 w_2 \Re \gamma\ip{\psi_2}{\psi_1} \label{eq:sl3}. 
\end{align}
Substituting Eqs.~\eqref{eq:sl1}, \eqref{eq:sl2}, and \eqref{eq:sl3} into Eq.~\eqref{eq:s_two_source}, we get Eq.~\eqref{eq:twosource_f}

\end{document}